%% file: sample-sigconf.tex
\documentclass[sigconf]{acmart}

\usepackage{booktabs} 
\usepackage{color, colortbl}
\usepackage{booktabs} 
\usepackage{graphicx}
\usepackage{courier}
\usepackage{makecell}
\usepackage{float}
\usepackage{bm}
\usepackage{multirow}
\usepackage[flushleft]{threeparttable}
\usepackage{flushend}
\usepackage{bbm}
\usepackage{mathrsfs}
\usepackage[ruled, linesnumbered]{algorithm2e} 
\SetKw{KwInput}{Input}
\SetKw{KwOutput}{Output}
\usepackage{bbm}
\usepackage{flushend}

\AtBeginDocument{%
  \providecommand\BibTeX{{%
    \normalfont B\kern-0.5em{\scshape i\kern-0.25em b}\kern-0.8em\TeX}}}

\newcommand{\ignore}[1]{}
\newcommand\ChangeRT[1]{\noalign{\hrule height #1}}
\newcommand{\ie}{\emph{i.e., }}

\newcommand{\etc}{\emph{etc}}
\newcommand{\etal}{\emph{et al. }}
\usepackage{etoolbox}
\usepackage{tabularx}
\usepackage{enumitem}
\usepackage{multirow}
\usepackage{balance}
\usepackage{breakurl}

\hypersetup{
   breaklinks=true,   
   colorlinks=true,   
   pdfusetitle=true,  
}

\copyrightyear{2023}
\acmYear{2023}
\setcopyright{rightsretained}
\acmConference[WSDM '23]{Proceedings of the Sixteenth ACM International Conference on Web Search and Data Mining}{February 27-March 3, 2023}{Singapore, Singapore}
\acmBooktitle{Proceedings of the Sixteenth ACM International Conference on Web Search and Data Mining (WSDM '23), February 27-March 3, 2023, Singapore, Singapore}\acmDOI{10.1145/3539597.3570395}
\acmISBN{978-1-4503-9407-9/23/02}

\makeatletter
\gdef\@copyrightpermission{
 \begin{minipage}{0.3\columnwidth}
  \href{https://creativecommons.org/licenses/by/4.0/}{\includegraphics[width=0.90\textwidth]{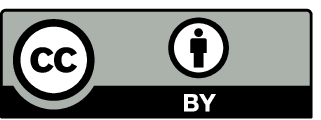}}
 \end{minipage}\hfill
 \begin{minipage}{0.7\columnwidth}
  \href{https://creativecommons.org/licenses/by/4.0/}{This work is licensed under a Creative Commons Attribution International 4.0 License.}
 \end{minipage}
 \vspace{5pt}
}
\makeatother

\begin{document}

\title{Model-based Unbiased Learning to Rank}

\author{Dan Luo}
\email{dal417@lehigh.edu}
\affiliation{%
  \institution{Lehigh University}
  \city{Bethlehem}
  \state{PA}
  \country{USA}
}

\author{Lixin Zou}
\email{zoulixin15@gmail.com}
\affiliation{%
  \institution{Baidu Inc.}
  \city{Beijing}
  \country{China}
}

\author{Qingyao Ai}
\email{aiqy@tsinghua.edu.cn}
\affiliation{%
  \institution{Tsinghua University}
  \city{Beijing}
  \country{China}
}

\author{Zhiyu Chen$^*$}
\email{zhiyuche@amazon.com}
\affiliation{%
  \institution{Amazon.com, Inc.}
  \city{Seattle}
  \state{WA}
  \country{USA}
}

\author{Dawei Yin$^{\dagger}$}
\email{yindawei@acm.org}
\affiliation{%
  \institution{Baidu Inc.}
  \city{Beijing}
  \country{China}
}

\author{Brian D. Davison}
\email{davison@cse.lehigh.edu}
\affiliation{%
  \institution{Lehigh University}
  \city{Bethlehem}
  \state{PA}
  \country{USA}
}

\renewcommand{\shortauthors}{Dan Luo et al.}

\thanks{
$^{\dagger}$ Corresponding author. \\
$^*$ The work was done prior to joining Amazon.
}


\begin{abstract}
Unbiased Learning to Rank~(ULTR), \ie learning to rank documents with biased user feedback data, is a well-known challenge in information retrieval. Existing methods in unbiased learning to rank typically rely on click modeling or inverse propensity weighting~(IPW). Unfortunately, search engines face the issue of a severe long-tail query distribution, which neither click modeling nor IPW handles well. Click modeling usually requires that the same query-document pair appears multiple times for reliable inference, which makes it fall short for tail queries; IPW suffers from high variance since it is highly sensitive to small propensity score values. Therefore, a general debiasing framework that works well under tail queries is sorely needed. To address this problem, we propose a model-based unbiased learning-to-rank framework. Specifically, we develop a general context-aware user simulator to generate pseudo clicks for unobserved ranked lists to train rankers, which addresses the data sparsity problem. In addition, considering the discrepancy between pseudo clicks and actual clicks, we take the observation of a ranked list as the treatment variable and further incorporate inverse propensity weighting with pseudo labels in a doubly robust way. The derived bias and variance indicate that the proposed model-based method is more robust than existing methods. Extensive experiments on benchmark datasets, including simulated datasets and real click logs, demonstrate that the proposed model-based method consistently outperforms state-of-the-art methods in various scenarios. 
The code is available at \url{https://github.com/rowedenny/MULTR}.

\end{abstract}

%
%


%

\maketitle

\input{1-intro}

\input{2-related_work}
\input{3-preliminary}

\input{4-method}

\input{5-exp}

\input{6-conclusion}

\bibliographystyle{ACM-Reference-Format}
\balance
\bibliography{sample-base}

\end{document}

%% file: 1-intro.tex
\section{Introduction}
Search engines serve as one of the most important tools for accessing information online. 
In modern search engines, learning to rank~(LTR) algorithms play a critical role by creating models to accurately order a list of candidate documents based on their relevance to the query. 
As deep supervised models have been widely applied and been state-of-the-art in many ranking tasks~\cite{DBLP:conf/sigir/DehghaniZSKC17, DBLP:conf/wsdm/ZamaniMSCT18, zou2022pre, zou2021pre, ye2022fast, chu2022h}, obtaining large-scale and high-quality training data has become a bottleneck for the development of large scale learning-to-rank systems~\cite{DBLP:conf/icml/2010ltr, zou2022large}. 
In practice, implicit feedback that reflects users' information needs~\cite{DBLP:journals/ftir/Sanderson10} provides natural, abundant sustainable training data for ranking optimization without costly time consumption and human annotation. 
Therefore, LTR with implicit feedback such as clicks has received considerable attention in the IR community. 


However, click data is biased since relevance is not the only factor influencing users’ clicks. 
For example, \textit{position bias} occurs because users are more likely to examine documents at higher ranks~\cite{DBLP:conf/wsdm/CraswellZTR08, DBLP:journals/tois/JoachimsGPHRG07}. Consequently, the highly ranked document may receive more clicks, and the relevant document may be perceived as a negative sample simply by not being examined by users. 
Furthermore, ranking positions~\cite{DBLP:conf/sigir/JoachimsGPHG05},  display differences~\cite{DBLP:conf/www/YuePR10, DBLP:conf/eccv/ShenZ14}, users’ first impression~\cite{DBLP:conf/www/OvaisiAZVZ20}, \etc., also influence the implicit feedback. These biases make the data deviate from reflecting true relevance, and jeopardize the learned ranking model's performance.


To be unaffected by biases~\cite{DBLP:conf/wsdm/JoachimsSS17},  there are two groups of methods: 
(1) \textbf{Click modeling methods} explicitly introduce an additional factor as bias, make hypotheses about users' browsing behaviors, and estimate true relevance by optimizing the likelihood of observed user clicks~\cite{DBLP:conf/www/ChapelleZ09, DBLP:conf/wsdm/CraswellZTR08, DBLP:conf/sigir/DupretP08, DBLP:conf/www/WangZDC13}. 
Click models are straightforward yet effective, and they have made promising progress for various applications, such as CTR prediction in live recommendation~\cite{DBLP:conf/recsys/GuoYLTZ19} and grid view web applications~\cite{DBLP:conf/www/ZhuangQWBQHC21}. 
However, search engines are faced with severe long-tail query distribution, where click models can fall short since multiple observations for the same query may not be available~\cite{DBLP:conf/sigir/AiBLGC18}. 
%
(2) \textbf{Propensity-based methods} treat the click bias as the counterfactual factor~\cite{rosenbaum1983central}, and re-weight the click data for a relevance-equivalent training loss through \textit{inverse propensity weighting} (IPW)~\cite{DBLP:conf/wsdm/JoachimsSS17,  DBLP:conf/sigir/WangBMN16}. 
To this end, many approaches~\cite{DBLP:conf/sigir/AiBLGC18,DBLP:conf/www/HuWPL19,DBLP:conf/cikm/VardasbiOR20} are proposed to jointly learn user bias models (i.e., propensity models) with unbiased rankers.
However, propensity-based methods generally suffer from high variance~\cite{DBLP:conf/sigir/Saito20}, which can lead to suboptimal estimation~\cite{DBLP:conf/nips/SwaminathanJ15}. 
Therefore, a general debiasing framework with low variance is needed.

\begin{figure}
    \centering
    \includegraphics[width=0.99\linewidth]{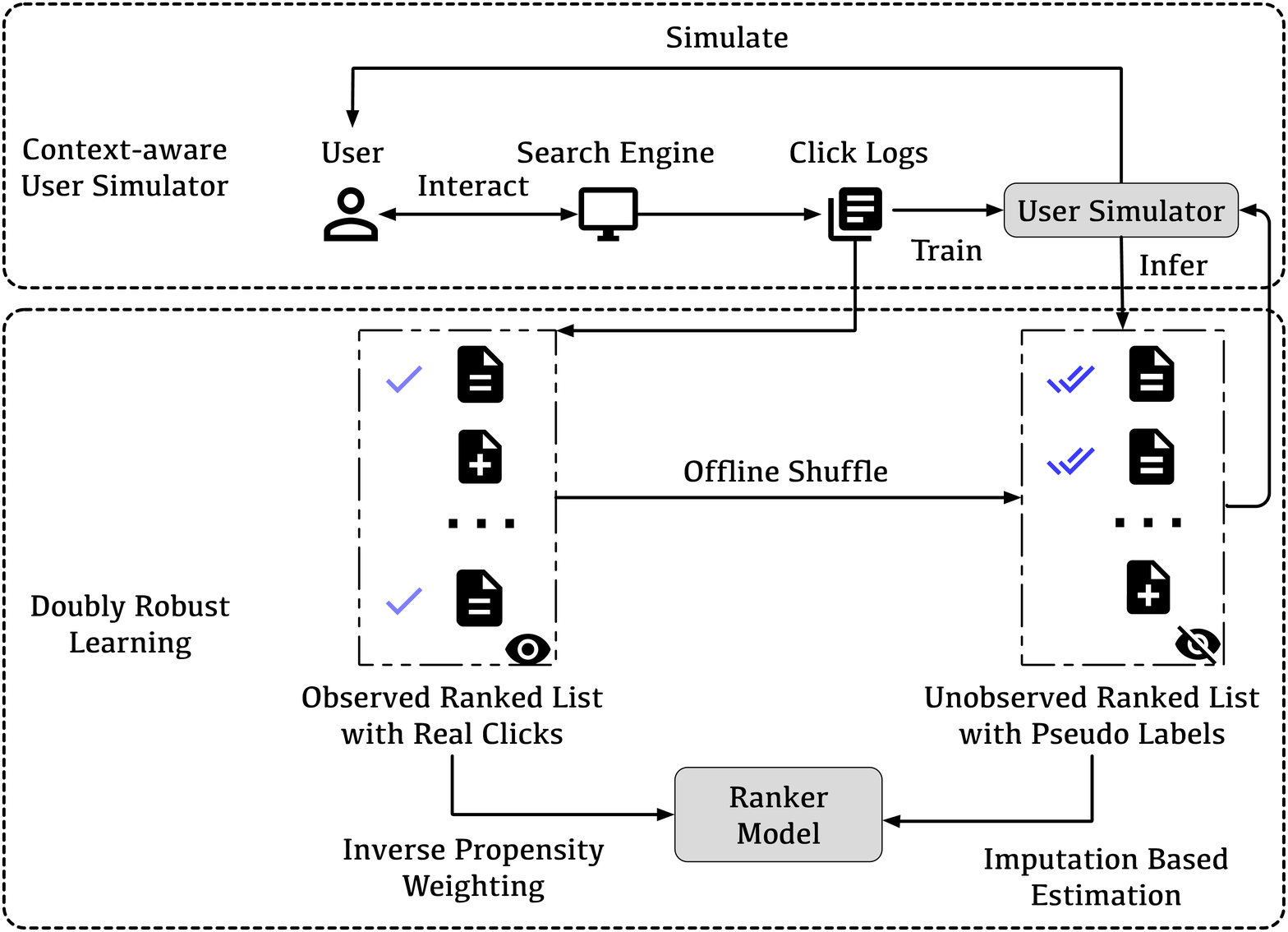}
    \caption{The framework of model-based unbiased learning to rank.}
    \label{fig:framework}
    \vspace{-.2in}
\end{figure}

To address these challenges, we propose a novel model-based unbiased learning to rank~(MULTR) framework, as shown in Figure~\ref{fig:framework}. 
MULTR consists of a context-aware user simulator for addressing the data sparsity problem and a doubly-robust learning algorithm for reducing the variance and achieving unbiased estimation.
Specifically, we first develop a user simulator that directly learns from click data, which generates pseudo clicks for unobserved ranked lists as data imputation and alleviates the sparse training data problem. Remarkably, the user simulator is unbiased since the factor, \ie the position, influencing the click is being considered as its input, which ensures the consistency between the training of user simulator and data augmentation.
However, as the user simulator cannot be as accurate as real users, the discrepancy between pseudo and actual clicks may mislead the ranking models, and unbiasedness cannot be guaranteed. 
To reduce the variance and ensure the unbiasedness, we resort to doubly robust approaches. 
In particular, we propose to take the \textit{observation of ranked lists} as the treatment, generate pseudo-click data 
via \textit{offline result randomization}, 
and incorporate the inverse propensity weighting with pseudo-click labels in a doubly robust way, where the user simulator essentially functions as the imputation model. 
The derived bias and variance indicate that the proposed model-based method has low bias and low variance. Thus, it is more robust than existing methods. To demonstrate the effectiveness of the proposed method, we conduct extensive experiments on simulated and real datasets, and demonstrate that the proposed model-based method consistently outperforms state-of-the-art methods.

The contributions of the proposed model-based unbiased learning to rank framework can be summarized as follows:
\begin{itemize}[leftmargin=*]
    \item We propose model-based unbiased learning to rank (MULTR) -- a general debiasing framework with low variance. 
    \item We devise a user simulator, which addresses data sparsity by generating pseudo-click labels for unobserved ranked lists.
    \item We propose the observation of rank lists as treatment, which addresses the treatment unobservability in unbiased learning to rank, and incorporate inverse propensity weighting with pseudo labels in a doubly robust way.
    \item We conduct extensive experiments on simulated datasets and real user click logs to demonstrate the superiority of MULTR. 
\end{itemize}



%% file: 2-related_work.tex
\section{Related Work}


\paragraph{Unbiased Learning to Rank}
To extract unbiased and reliable relevance signals from biased click signals, there are two streams of unbiased learning to rank methodologies. One school depends on click modeling, which makes assumptions about user browsing behaviors. Such methods maximize the likelihood of the observed data, model the examination probability and infer accurate
relevance feedback from user clicks~\cite{DBLP:conf/wsdm/CraswellZTR08, DBLP:conf/cikm/ChapelleMZG09, DBLP:conf/wsdm/ChuklinMR16, DBLP:MaoLZM18}. For example, \citet{DBLP:conf/recsys/GuoYLTZ19} propose a position-bias aware learning framework for CTR prediction that models position-bias in offline training and conducts online inference without biased information. \citet{DBLP:conf/wsdm/WangGBMN18} unify the training of the ranker model and the estimation of examination propensity with a graphical model and an EM algorithm. 
Despite their success, one major drawback of click models is that they usually require that the same query-document pair appears multiple times for reliable inference~\cite{DBLP:conf/ictir/MaoCLZM19}; thus they may fall short for tail queries.
The other school derives from counterfactual learning, which treats bias as a counterfactual factor and debiases user clicks via  inverse propensity weighting~\cite{DBLP:conf/wsdm/JoachimsSS17, DBLP:conf/sigir/WangBMN16}. Recent efforts jointly model the propensity estimation and unbiased learning to rank. For instance, \citet{DBLP:conf/sigir/AiBLGC18} and \citet{DBLP:conf/www/HuWPL19} present dual learning frameworks for estimating bias and training a ranking model. \citet{DBLP:conf/sigir/VardasbiRM20} propose cascade model-based inverse propensity scoring for propensity estimation in the cascade scenario. 
However, these methods ignore the severe high variance problem in IPW-based methods, particularly for long-tail data.
Our paper falls in the same propensity-based framework; however, we address the high variance problem by optimizing the ranking model in a doubly robust way, which has low bias and low variance.

\paragraph{Model-based Methods}
Model-based methods aim to construct a predictive user model and ask the question of the counterfactual form ``what will the user click if the search result page is presented differently?'', which naturally fits debiasing. Recent work has proposed some methods that improve ranking performance based on user simulators. For example, ~\citet{DBLP:conf/cikm/DaiHLXT0H0020} propose a utility estimator to generate counterfactual data to train the ranking model. ~\citet{DBLP:conf/cikm/ZhangMLZ0MXT19} develop a reinforcement learning algorithm to learn ranking policies in the simulation environments. While both methods show promising improvements, the theoretical guarantees with respect to unbiasedness are unclear. Different from such work, we leverage the user simulator as the imputation model in a doubly-robust way, which addresses the high bias issue caused by the discrepancy between user simulator and real clicks. More importantly, our derivation demonstrates the unbiasedness of the proposal. Our ablation study verifies that the doubly robust approach can effectively leverage both real and pseudo clicks, and improve the ranking performance.

\paragraph{Doubly-Robust Methods}
Doubly-Robust (DR) methods have been widely applied to position-biased clicks. \citet{DBLP:conf/recsys/Saito20} proposes a DR method for post-click conversions. ~\citet{DBLP:conf/sigir/GuoZLYCWCY021} further develop a more doubly robust estimator to reduce the variance. \citet{DBLP:conf/wsdm/KiyoharaSMNSY22} develop a cascade doubly robust estimator for off-policy evaluation, \citet{DBLP:conf/cikm/YuanHYZCDL19} introduce a DR estimator for click-through-rate prediction. \citet{DBLP:conf/cikm/ZouHCWCCYGY22} propose a doubly robust estimator for relevance estimation. A significant difference between current DR estimators and those in ULTR is that existing solutions use corrections based on action propensities, which is similar to generic counterfactual estimation; while ULTR needs examination propensities~\cite{oosterhuis2022doubly}, which unfortunately is unobservable in click logs.
In this paper, we address this challenge by taking the observation of rank lists as treatment, and incorporate the inverse propensity weighting with pseudo-click labels in a doubly-robust way.

%% file: 3-preliminary.tex
\section{Preliminaries}

This section presents the task of unbiased learning to rank and reviews the preliminary methods of click modeling and propensity-based methods with their strengths and weaknesses.

\subsection{Task Formulation}
Let $\mathcal{D}$ be the universal set of documents, and $\mathcal{Q}$ be the universal set of queries. For a user-issued query $q \in \mathcal{Q}$, we use $r_{d}$ to denote the relevance annotation over the query $q$ and document $d \in \mathcal{D}$. The goal of learning to rank is to find a mapping function from a query document pair $(q, d)$ to its relevance $r_d$ as $f:\mathcal{Q}\times\mathcal{D}\rightarrow\mathbb{R}$. 
A corresponding local loss function is usually proposed to learn the best $f$, given its retrieved ranked list $\pi_q$ as,  
\begin{equation}
    \label{eq:ideal_loss}
    \ell_{ideal}(f,q|\pi_q) = \sum_{d\in\mathcal{D}}\Delta(f(q,d),r_d|\pi_q),
\end{equation}
where $\Delta$ is a function that computes the individual loss on each document. $\ell_{ideal}(f, q)$ is the ideal local ranking loss for optimizing the ranking function $f$ with all the documents annotated. Without loss of generality, we simplify the individual loss function $\Delta(f(q,d), r_{d} | \pi_q)$ as  $\Delta(r_d|\pi_q)$. 

Typically, the relevance annotation $r_d$ is elicited through expert judgment; thus $r_d$ is considered to be unbiased, but expensive. An alternative approach is to use click data as relevance feedback from users. Suppose there is a click dataset in which the clicks on documents with respect to queries by an initial ranking model are logged. If we conduct learning to rank by replacing the relevance label $r_d$ with click label $c_d$ in Equation~\ref{eq:ideal_loss}, then the empirical local ranking loss is derived as follows, 
\begin{eqnarray}
    \label{eq:loss_naive}
    \ell_{naive}(f,q|\pi_o) = \sum_{d \in \pi_o, c_{d}=1} \Delta(c_{d}|\pi_o),
\end{eqnarray}
where $\pi_o$ is the \textit{observed} ranked list, which is the ranked list presented to users, and $c_{d}$ is a binary variable indicating whether the document $d$ in the ranked list $\pi_o$ is clicked.
However, this naive loss function is biased. For instance, position bias occurs because users are more likely to examine the documents at higher ranks~\cite{DBLP:conf/wsdm/JoachimsSS17}. Consequently, highly ranked documents may receive more clicks, and relevant (but unclicked) documents may be perceived as negative samples because they are unexamined by users. To address this issue, unbiased learning-to-rank aims to eliminate bias in click data and then train a ranking model with the resulting user clicks. 

\subsection{Click Modeling}
To eliminate biases in click data, one intuitive method is to model bias and relevance as independent factors, and represent the joint effects of bias and relevance with bias-aware click modeling~\cite{DBLP:conf/recsys/GuoYLTZ19, DBLP:conf/sigir/LiuCDHP020}. Here we refer to the general approach as click modeling. Let $\mathbf{b}$ be the bias features, and $\mathbf{x}_d$ be the vector representation of document $d$, then the bias-aware click predictor is modeled as follows:
\begin{equation}
    \hat{P}(c_d=1, q) = h(\mathbf{b}; \phi) \oplus f(\mathbf{x}_d; \theta),
\end{equation}
where $\hat{P}(c_d=1, q)$ is the estimated click probability of document $d$ of query $q$. $h$ is the bias model parameterized by $\phi$ that generates a score based on the inputs of bias features. $f$ is the relevance prediction model parameterized by $\theta$ that generates a relevance score based on the relevance representations. $\oplus$ is an operation that combines the bias-based score and the relevance score, which could be addition~\cite{DBLP:conf/www/ZhuangQWBQHC21}, multiplication~\cite{DBLP:conf/recsys/GuoYLTZ19}, \etc. The general local ranking loss of a bias-aware click predictor is defined as:
\begin{equation}
    \begin{aligned}
        \ell_{CLICK} (f, h, q | \pi_o) 
        &=  \sum_{d \in \pi_o, c_d=1} \mathrm{CE} (\hat{P}(c_d=1, q), c_d) \\
        &= \sum_{d \in \pi_o, c_d=1} \mathrm{CE} \left( h(\mathbf{b}; \phi) \oplus f(\mathbf{x}_d; \theta), c_d \right),
    \end{aligned}
\end{equation}%
where $c_d$ is a binary variable indicating whether document $d$ is clicked, and $\mathrm{CE}  (\cdot, \cdot)$ the cross-entropy function. When conducting the relevance inference, one can easily eliminate the influence of bias by dropping the scores from $h(\mathbf{b}; \phi)$, \ie directly ranking with $f(\mathbf{x}_d, \theta)$. Though click models work well with head queries, they could fall short when multiple observations of the same query may not be available~\cite{DBLP:conf/sigir/AiBLGC18}.

\subsection{Propensity-based Methods}
Inverse propensity weighting (IPW) is the first unbiased learning to rank algorithm proposed under the propensity-based framework~\cite{DBLP:conf/sigir/WangBMN16, DBLP:conf/wsdm/JoachimsSS17}.  Let $e_d, c_d$ be the binary variables that represent whether document $d$ is examined and clicked by a user, based on the \textit{Examination Hypothesis} \cite{DBLP:conf/www/RichardsonDR07} that a user would only click a document when it is observed by the user and considered relevant to the user’s need, we have 
\begin{equation}
    \label{eq:assumption}
    c_{d}= 1 \Longleftrightarrow (e_{d}= 1 \ \ \text{and} \ \ r_{d}=1).
\end{equation} %
Then, IPW instantiates the local ranking loss as
\begin{equation}
    \label{eq:ipw_loss}
    \ell_{IPW} (f, q|\pi_o) = \sum_{d \in \pi_o, c_{d}=1} \frac{\Delta (c_{d}|\pi_o)}{\hat{P}(e_{d} = 1)},
\end{equation}%
where $\hat{P}(e_{d} = 1)$ is the estimated probability that document $d$ is examined in the query session.  A nice property is that only clicked documents $c_d = 1$ contribute to the estimation in Eq.~\ref{eq:ipw_loss}.

\paragraph{Bias and Variance Analysis}  
We analyze the bias and variance of the inverse propensity weighting approach. 
\begin{theorem}
    \label{th:ips_bias}
    The bias of the IPW estimators is
    \begin{equation}
        \begin{aligned}
            {Bias}[\ell_{IPW} (f, q|\pi_o)] = \left| \sum_{d \in \pi_o, r_{d}=1} \frac{P(e_{d}=1) - \hat{P}(e_{d}=1)}{\hat{P}(e_{d} = 1)} \Delta (c_{d}|\pi_o) \right|
        \end{aligned}
    \end{equation}
\end{theorem}
\begin{proof} 
~\citet{DBLP:conf/wsdm/JoachimsSS17} proved that a ranking model trained with clicks and IPW loss will converge to the same model trained with true relevance labels; see ~\citet{DBLP:conf/sigir/AiBLGC18} for more detail.
\end{proof}
As shown in Theorem~\ref{th:ips_bias},  when the estimated examination probability is accurate, \ie $P(e_{d_i}=1) = \hat{P}(e_{d_i}=1)$, then $Bias[\ell_{IPW}] = 0$, indicating accurate estimation can be achieved with rich observation data. Then, we derive the variance of the IPW estimators.

%
\begin{theorem} The variance of the IPW estimator is
    \label{th:ips_var}
    \begin{equation}
        \small{
        \begin{aligned}
            \mathbb{V}_{\mathcal{O}} [\ell_{IPW} (f, q|\pi_o)] = \sum_{d \in \pi_o, c_d=1} \frac{P(e_{d} = 1)(1 - P(e_{d} = 1))}{\hat{P}(e_{d} = 1)^2} \Delta (c_{d}|\pi_o)^2.
        \end{aligned}}
    \end{equation}
\end{theorem}
\begin{proof}
For the IPW estimator, its variance on the observed ranking documents is:
    \begin{equation}
    \label{eq:ips_var}
        \small{
        \begin{aligned}
            \mathbb{V}_{\mathcal{O}} [\ell_{IPW} (f, q|\pi_o)] 
            =& \mathbb{V}_{\mathcal{O}} \left [ \sum_{d \in \pi_o, c_{d}=1} \frac{\Delta (c_{d}|\pi_o)}{\hat{P}(e_{d} = 1)} \right] \\
            =&  \sum_{d \in \pi_o, c_{d}=1} \mathbb{V}_{\mathcal{O}} \left [ \frac{\Delta (c_{d} | \pi_o)}{\hat{P}(e_{d} = 1)} \right] \\
            =&  \sum_{d \in \pi_o, c_{d}=1} \mathbb{V}_{\mathcal{O}} [e_{d}] \cdot  \left( \frac{\Delta (c_{d} |\pi_o)}{\hat{P}(e_{d} = 1)}  \right)^2 \\
            =& \sum_{d \in \pi_o, c_{d}=1} \left( \mathbb{E}_{\mathcal{O}} [e^2_{d}]  - \mathbb{E}^2_{\mathcal{O}} [e_{d}] \right)
            \cdot  \left( \frac{\Delta (c_{d} | \pi_o)}{\hat{P}(e_{d} = 1)}  \right)^2 \\
            =& \sum_{d \in \pi_o, c_d=1} \frac{P(e_{d} = 1)(1 - P(e_{d} = 1))}{\hat{P}(e_{d} = 1)^2}
            \cdot   \Delta (c_{d}|\pi_o)^2.       
        \end{aligned}}
    \end{equation}
 \end{proof}
Theorem~\ref{th:ips_var} illustrates that the variance of IPW estimators depends on the estimated propensity. When $\hat{P}(e_{d}=1)$ is small, it may lead to high variance. Especially for  long-tail queries with rare observation data, the high variance of IPW becomes a problem that directly influences the effectiveness of the ranker.




%% file: 4-method.tex
\section{Model-based Unbiased Learning to Rank}
In this section, we first introduce the context-aware user simulator, which generates pseudo-click labels for unobserved ranked lists as data imputation and addresses the data sparsity.  
Afterwards, we propose a doubly robust estimator, which takes the observation of ranked lists as the treatment, and further incorporates inverse propensity weighting with pseudo labels from the user simulator above in a doubly robust way.
Lastly, we derive the bias and variance of the proposed method, and demonstrate that our proposed method is more robust than existing methods.

\subsection{Context-aware User Simulator}
The user simulator consists of two important components: a local context encoder, which captures different feature distributions from different queries, and a context-aware click decoder, which produces the click probability of each document sequentially. The framework of our user simulator is shown in Fig.~\ref{fig:user_simulator}.

\begin{figure}
    \centering
    \includegraphics[width=3.3in]{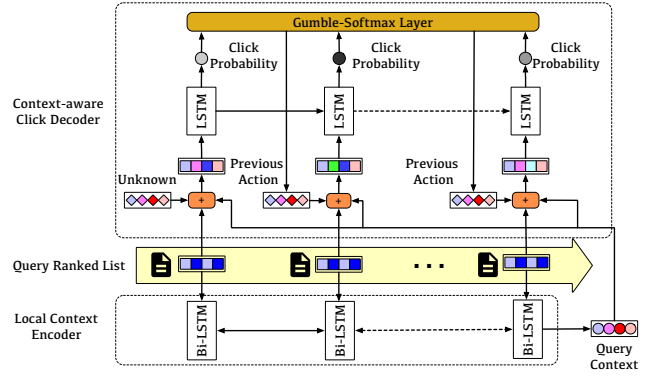}
    \vspace{-10pt}
    \caption{Framework of the Context-aware User Simulator. }
    \label{fig:user_simulator}
    \vspace{-10pt}
\end{figure}

\paragraph{Local Context Encoder}
To capture the characteristics of queries, \ie different queries may have different distributions in the feature space~\cite{DBLP:conf/sigir/AiBGC18}, we first design a local context encoder for each query and use it to refine the query-specific features, as shown at the bottom block in Figure~\ref{fig:user_simulator}. Formally, given the top $N$ documents $\{d_i\}_{i=1}^N$ in a ranked list $\pi_q$ from top to bottom, and their corresponding feature vectors $\{ \mathbf{x}_i \}_{i=1}^N$, we use a bi-directional long short-term memory~\cite{DBLP:journals/neco/HochreiterS97} network to obtain a contextual representation of each document with respect to the entire ranked list:
\begin{equation}
    \begin{aligned}
         \mathbf{h}_i &= [\overrightarrow{\text{LSTM}} (\mathbf{x}_i, i); \overleftarrow{\text{LSTM}} (\mathbf{x}_i, i)] \\
         \mathbf{h}_{\pi_q} &= \mathbf{W}^\top \mathbf{h}_N,
    \end{aligned}
\end{equation}
where $[\cdot; \cdot]$ is concatenation, and $\overrightarrow{\text{LSTM}} (\mathbf{x}_i, i)$ processes the document $x$  from top to bottom and returns the LSTM hidden state at position $i$ (and vice versa for the backward direction $\overleftarrow{\text{LSTM}}$). We then produce the local context representation $\mathbf{h}_{\pi_q}$ by a linear transformation of the last hidden state $\mathbf{h}_N$. 


\paragraph{Context-aware Click Decoder}

To characterize different influences from the previous actions, we encode actions with an embedding matrix $\mathbf{A}$. In our settings, there are three types of actions, \ie $\textsf{click}$, $\textsf{skip}$, and $\textsf{unknown}$, where $\textsf{unknown}$ is designed for the initial step. We use  $\mathbf{h}_{a_t}$ to denote the action embedding corresponding to action $a_t$ at step $t$. The context-aware click decoder is formulated as:
\begin{equation}
    \label{eq:click_decoder}
    \begin{split}
        \mathbf{z}_{\pi} &= \mathbf{W}_1^\top \mathbf{h}_{\pi_q} + \mathbf{b}_1 \\
        \mathbf{z}_{x_i} &= \mathbf{W}_2^\top \mathbf{x}_i + \mathbf{b}_2 \\
        \mathbf{z}_{a_t} &= \mathbf{W}_3^\top \mathbf{h}_{a_t} + \mathbf{b}_3 \\
        \mathbf{z}_t &= [\mathbf{z}_{\pi}; \mathbf{z}_{x_i}; \mathbf{z}_{a_t}] \\
        \mathbf{h}_t &= \overrightarrow{\text{LSTM}} (\mathbf{z}_t, t), \\
    \end{split}
\end{equation}%
where $\mathbf{W}_1, \mathbf{W}_2, \mathbf{W}_3, \mathbf{b}_1, \mathbf{b}_2, \mathbf{b}_3$ are the trainable parameters used to transform the query context vector, the document vector and the previous action embedding, respectively.  

At each step $t$, the hidden state $\mathbf{h}_t$ is projected into a conditional click probability score through the sigmoid function:
\begin{equation}
    \label{eq:click_prob}
    p_{t} = \mathrm{sigmoid} (\mathbf{W}_4^\top \mathbf{h}_t + b_4),
\end{equation}%
where $\mathbf{W}_4, \mathbf{b}_4$ are trainable parameters, and $p_{t} \in [0, 1]$ represents the probability that a user clicks document $d_t$.


\paragraph{User Simulator Optimization}
After obtaining the click probability score $p_t$ for each document, we train the user simulator by applying the binary cross-entropy:
\begin{equation}
    \label{eq: user_simulator_loss}
    \ell_{e}(g, q |\pi_o) = \sum_{t =1}^N \left( - c_{t} \log p_{t} - (1 - c_{t}) \log(1 - p_{t}) \right) + \lambda ||\phi||^2, 
\end{equation}
where $g$ is the whole context-aware user simulator, $c_{t}$ denotes the click label of document $d_t$ in observed ranked list $\pi_o$,  $\phi$ denotes all the parameters of the user simulator $g$, and $\lambda$ denotes the $L_2$ regularization coefficient. 


\textbf{Discussion}. It is worth noting that the goal of this paper is not to propose a high-performance user simulator. More sophisticated models, such as utility estimator~\cite{DBLP:conf/cikm/DaiHLXT0H0020} and UBS4RL~\cite{10.1145/3511469}, can be applied here. We leverage the user simulator as the imputation model to generate pseudo-click labels for an arbitrary ranked list. In the next section, we will leverage the user simulator in a doubly robust way to obtain unbiased ranking models.


\subsection{Doubly Robust Learning}
User simulators have the benefit of data imputation~\cite{DBLP:conf/icml/Hernandez-LobatoHG14a}. However, directly optimizing the ranking model through the user simulator could be highly biased due to the discrepancy between pseudo labels and actual clicks, \ie the pseudo clicks cannot be as accurate as real clicks. To address this problem,  we resort to the doubly robust approach.
Unfortunately, the treatment, \ie document examination, is not directly observable in click data. It is because when a document is not clicked, we cannot determine whether the user chose not to click or the user did not examine it~\cite{oosterhuis2022doubly}. 
To overcome this bottleneck, we propose to take the \textit{observation of ranked lists} as the treatment. After that, we incorporate inverse propensity weighting with pseudo-clicked labels in a doubly robust way, where the user simulator essentially functions as the imputation model.

Before we describe our doubly robust estimator for unbiased learning to rank, we introduce two concepts: prediction error and imputation error. We refer to the local ranking loss $\ell_{naive}$ in Eq~\ref{eq:loss_naive} with actual clicks as \textit{prediction error}. Then the \textit{imputed error} $\ell_{IMP}$, \ie the estimated value of the prediction error, is defined as: 
\begin{equation}
    \label{eq:imputed_error}
    \ell_{IMP} (f, q|\pi_o) = \sum_{d \in \pi_o, \hat{c}_{d} =1}  \Delta(\hat{c}_{d} | \pi_o),
\end{equation}
where $\hat{c}_{d}$ is the imputed click generated from the user simulator for each document $d$ in the observed ranked list $\pi_o$ under query $q$.

Let $\Pi_q$ be the set of document permutations, \ie all possible ranked lists. We simplify $\Pi_q$ as $\Pi$ when there is no ambiguity for the issued query $q$. 
We propose our doubly robust (DR) estimator by combining the imputed errors for all ranked lists, and the IPW-based prediction error for observed ranked lists. The loss function of the DR estimator given query $q$ is defined as:
\begin{equation}
\small{
\label{eq: dr_loss}
\begin{aligned}
\ell_{DR}(f, q|\pi_o) &= \frac{1}{|\Pi|} \sum_{\pi \in \Pi} \Bigg\{ \hat{\ell}_{IMP} (f, q|\pi) \\
    & \quad + o_{\pi} \bigg[ \ell_{IPW}(f,q|\pi) - \hat{\ell}_{IPW}(f,q|\pi)\bigg] \Bigg\} \\
    &= \frac{1}{|\Pi|} \sum_{\pi \in \Pi} \Bigg\{ \Bigg. \sum_{d \in \pi, \hat{c}_{d} =1}  \Delta(\hat{c}_{d}|\pi) \\
    & \quad + o_{\pi}\bigg[ \sum_{d \in \pi, c_{d}=1} \frac{\Delta (c_{d}|\pi)}{\hat{P}(e_{d} = 1)} - \sum_{d \in \pi, \hat{c}_{d}=1} \frac{\Delta (\hat{c}_{d}|\pi)}{\hat{P}(e_{d} = 1)} \bigg] \Bigg. \Bigg\},
\end{aligned}
}
\end{equation} %
where $\hat{\ell}_{IPW}(f,q|\pi)= \sum_{d \in \pi, \hat{c}_{d} =1}  \frac{\Delta(\hat{c}_{d} | \pi)}{\hat{P}(e_{d} = 1)} $ is the propensity score weighted imputed error, and $o_{\pi} =\mathbbm{1}\{\pi=\pi_o\}$ indicates whether $\pi$ is an observed ranked list. 
A nice property of the DR estimator is: if either the imputed error of any unobserved ranked list or the estimated propensities of any observed ranked list is accurate, the DR estimator is then unbiased~\cite{DBLP:conf/sigir/GuoZLYCWCY021,DBLP:conf/cikm/ZouHCWCCYGY22}, which is recognized as \textit{double robustness}.

\textbf{Discussion}.
Careful readers may notice that our method is similar in spirit to result randomization.
Ideally, if the pseudo-click labels generated from the user simulator are accurate for any unobserved ranked lists, learning from user simulators would be equivalent to learning from online result randomization, which is unbiased in principle~\cite{DBLP:conf/wsdm/JoachimsSS17, DBLP:conf/sigir/WangBMN16}. 
However, online result randomization is often impractical since it can negatively affect the user experience. Our method overcomes the aforementioned limitations by leveraging the user simulator in a doubly robust way, and simultaneously obtains unbiased ranking models.

\subsection{Bias and Variance of DR Estimator}

In this section, we derive the bias and variance of MULTR, and prove its double robustness. 

\begin{theorem} 
    \label{th: dr_bias}
    Let $\delta_{\pi, d} = c_{d} \Delta(c_{d}|\pi) - \hat{c}_{d} \Delta(\hat{c}_{d}|\pi)$ be the error deviation, and $\rho_{q, {d}} = \frac{P(e_{d}=1) - \hat{P}(e_{d}=1)}{\hat{P}(e_{d}=1)} $ be the propensity deviation. The bias of the doubly robust (DR) estimator is 
    \begin{equation}
        \begin{aligned}
            Bias[\ell_{DR} (f, q|\pi_o)] = \left| \frac{1}{|\Pi|} \sum_{\pi \in \Pi} \sum_{d \in \pi}  \rho_{\pi, d} \delta_{\pi, d} \right| .
        \end{aligned}
    \end{equation}
\end{theorem}

\begin{proof} According to the definition of bias,
\begin{equation}
\small{
    \begin{aligned}
         & Bias[\ell_{DR} (f, q|\pi_o)] \\
        =& \Bigg | \mathbb{E}_{\mathcal{O}} \left[ \ell_{DR} (f, q|\pi_o) \right ] -  \ell_{ideal}(f,q|\pi_o) \Bigg | \\
        =& \Bigg | \mathbb{E}_{\mathcal{O}} \left[ \ell_{DR} (f, q|\pi_o) \right ] -  \sum_{d \in \pi, r_{d}=1} \Delta(c_{d}|\pi) \Bigg | \\
        =& \Bigg | \mathbb{E}_{\mathcal{O}} \left[ \ell_{DR} (f, q|\pi_o) \right ] -  \frac{1}{|\Pi|}\sum_{\pi \in \Pi}\sum_{d \in \pi, c_{d}=1} \Delta(c_{d}|\pi) \Bigg | \\
        =& \Bigg| \Bigg. \frac{1}{|\Pi|} \sum_{\pi \in \Pi}   \Bigg.  \mathbb{E}_{\mathcal{O}} \bigg\{  \big[ \sum_{d \in \pi, \hat{c}_{d} =1}  \Delta(\hat{c}_{d}|\pi) - \sum_{d \in \pi, c_{d}=1} \Delta(c_{d}|\pi_o) \big] \\
        & \quad + o_{\pi} \big[ \sum_{d \in \pi, c_{d}=1} \frac{\Delta (c_{d}|\pi)}{\hat{P}(e_{d} = 1)} - \sum_{d \in \pi, \hat{c}_{d}=1} \frac{\Delta (\hat{c}_{d}|\pi)}{\hat{P}(e_{d} = 1)} \big] \bigg\} \Bigg. \Bigg| \Bigg.     
    \end{aligned}
}
\end{equation} %
In the third line, we expand $\ell_{ideal}$ with online result randomization. 
Based on the examination hypothesis in Eq~\ref{eq:assumption}, the terms in the first row can be derived as, 
\begin{equation}
    \begin{aligned}
        & \left |  \frac{1}{|\Pi|} \sum_{\pi \in \Pi} \mathbb{E}_{\mathcal{O}} \left[ \sum_{d \in \pi, \hat{c}_{d} =1}  \Delta(\hat{c}_{d}|\pi) - \sum_{d \in \pi, c_{d}=1} \Delta(c_{d}|\pi_o) \right] \right | \\
        =&  \left| \frac{1}{|\Pi|} \sum_{\pi \in \Pi}  
         \sum_{d \in \pi}  \left[ \hat{c}_d \cdot \Delta(\hat{c}_{d}|\pi) - c_d \cdot \Delta(c_{d}|\pi) \right] \right|. 
    \end{aligned}
\end{equation}
Then the second term can be derived as, 
\begin{equation}
\small{
\begin{aligned}
    & \left| \frac{1}{|\Pi|} \sum_{\pi \in \Pi} \mathbb{E}_{\mathcal{O}}  \left \{  o_{\pi} \left[ \sum_{d \in \pi, c_{d}=1} \frac{\Delta (c_{d}|\pi)}{\hat{P}(e_{d} = 1)} - \sum_{d \in \pi, \hat{c}_{d}=1} \frac{\Delta (\hat{c}_{d}|\pi)}{\hat{P}(e_{d} = 1)} \right] \right \} \right| \\
    =& \left| \frac{1}{|\Pi|} \sum_{\pi \in \Pi} \mathbb{E}_{\mathcal{O}}  \left \{  o_{\pi} \left[ \sum_{d \in \pi, c_{d}=1} \frac{\Delta (c_{d}|\pi)}{\hat{P}(e_{d} = 1)} \right ] \right \} - \mathbb{E}_{\mathcal{O}}  \left \{ o_{\pi} \left [ \sum_{d \in \pi, \hat{c}_{d}=1} \frac{\Delta (\hat{c}_{d}|\pi)}{\hat{P}(e_{d} = 1)} \right] \right \} \right| \\
    =& \left| \frac{1}{|\Pi|} \sum_{\pi \in \Pi} 
    \sum_{d \in \pi} \frac{P(e_{d} = 1)}{\hat{P}(e_{d} = 1)} \left [ c_d \cdot \Delta (c_{d}|\pi) - \hat{c}_d \cdot \Delta (\hat{c}_{d}|\pi) \right ] \right| .
\end{aligned}}
\end{equation}
Finally, combing the two parts above, we derive the bias of the proposed DR estimator 
\begin{equation}
\begin{aligned}
    Bias[\ell_{DR} (f, q|\pi_o)]  
    = \left| \frac{1}{|\Pi|} \sum_{\pi \in \Pi} \sum_{d \in \pi }  \rho_{q, d} \delta_{q, d} \right| .
\end{aligned}
\end{equation}
\end{proof} %
As shown in Theorem ~\ref{th: dr_bias}, when the imputed error is accurate $(\delta_{q, d} = 0)$ or the propensity estimation is accurate $(\rho_{q,d} = 0)$, the bias of DR estimator $Bias \left[ \ell_{DR} (f, q | \pi_o)  \right] = 0$. 
Next, we derive the variance of the DR estimator.

\begin{theorem}
\label{th:dr_vairance}
    The variance of the DR estimator is 

\begin{equation}
\small{
\begin{aligned}
    \mathbb{V}_{\mathcal{O}} [\ell_{DR} (f, q|\pi_o)]  = \frac{1}{|\Pi|^2} \sum_{\pi \in \Pi} \sum_{d \in \pi} \bigg\{ \Bigg. \frac{P(e_{d} = 1) \left( 1 - P(e_{d} = 1) \right)}{\hat{P}(e_{d} = 1)^2} \cdot   \delta^2_{q, d} \bigg\} \Bigg.
\end{aligned}}
\end{equation}
\end{theorem}

\begin{proof}
The variance of $\ell_{DR} (f, q|\pi_o)$ with respect to the observation indicator $\mathcal{O}$ is
\begin{equation}
\small{
\begin{aligned}
    & \mathbb{V}_{\mathcal{O}}[\ell_{DR} (f, q|\pi_o)] \\
    =& \frac{1} {|\Pi|^2} \sum_{\pi \in \Pi} \mathbb{V}_{\mathcal{O}}  \Bigg\{ \hat{\ell}_{IMP} (f, q|\pi) + \mathcal{O} \bigg[ \ell_{IPW}(f,q|\pi) - \hat{\ell}_{IPW}(f,q|\pi)\bigg] \Bigg\} \\
    =& \frac{1} {|\Pi|^2} \sum_{\pi \in \Pi} \mathbb{V}_{\mathcal{O}} [o_{\pi}] \cdot \bigg[ \ell_{IPW}(f,q|\pi) - \hat{\ell}_{IPW}(f,q|\pi)\bigg]^2 \\
    =& \frac{1} {|\Pi|^2} \sum_{\pi \in \Pi}  \left\{ \mathbb{V}_{\mathcal{O}} [o_{\pi}] \cdot \left [ \sum_{d \in \pi} \left( \frac{c_{d} \cdot \Delta (c_{d}|\pi)}{\hat{P}(e_{d} = 1)} - \frac{\hat{c}_{d} \cdot  \Delta(\hat{c}_{d}|\pi)}{\hat{P}(e_{d} = 1)} \right) \right ]^2 \right\} \\
    =& \frac{1} {|\Pi|^2} \sum_{\pi \in \Pi} \left\{ \frac{P(e_{d} = 1) \left( 1 - P(e_{d} = 1) \right)}{\hat{P}(e_{d} = 1)^2} \cdot   \delta^2_{q, d} \right\},
\end{aligned}
}
\end{equation}
where the third step utilizes the similar derivation in Eq~\ref{eq:ips_var}. 


\end{proof}
In Theorem~\ref{th:dr_vairance}, we demonstrate that the variance of the DR estimator depends on the estimated propensity, \ie $\hat{P}(e_{d} = 1)$, which may lead to a high variance problem. Recall the variance of the IPW estimator in Theorem~\ref{th:ips_var}; it is worth noting that our DR estimator can still reduce the variance of the IPW estimator, if any given observed ranked list satisfies $\left[ c_d \Delta (c_{d}|\pi) - \hat{c}_d \Delta (\hat{c}_{d} | \pi) \right]^2 \leq [c_d  \Delta (c_{d} |\pi)]^2$. Then, the variance of the DR estimator is smaller than that of the IPW-based estimator.

\subsection{Learning Framework}
We summarize the learning framework in Algorithm~\ref{alg:mutlr}.
It is worth noting that the proposed MULTR is a plug-in model, which can be seamlessly integrated into any IPW-based unbiased learning to rank framework without result randomization, such as DLA~\cite{DBLP:conf/sigir/AiBLGC18} and  REM~\cite{DBLP:conf/wsdm/WangGBMN18}. 

\begin{algorithm}
\caption{Model-based Unbiased Learning to Rank}\label{alg:mutlr}
\SetKwInOut{Input}{Input}
\SetKwInOut{Output}{Output}

\Input{query set $\mathcal{Q}$, an IPW-based framework $F$}
\Output{ranking model $f(\theta)$ and  user simulator $g(\phi)$}

Initialize the parameters $\theta, \phi$ \;
\For{number of steps for \textbf{training the user simulator}} {
    Sample a batch of queries $Q$ from $\mathcal{Q}$, with observed ranked lists $\Pi_o$ \;
    Update $\phi$, according to Eq. \ref{eq: user_simulator_loss} \;
}

\For{number of steps for \textbf{training the unbiased ranking model}} {
    Sample a batch of queries $Q$ from $\mathcal{Q}$, with observed ranked lists $\Pi_o$ \;
    Obtain unobserved ranked lists $\Pi_u$ by result randomization over $\Pi_o$ \;
    Generate pseudo-click labels $\hat{c}_d$ for all ranked list $\Pi = \Pi_o \cup \Pi_u$ by sampling on the click probabilities, according to Eq.~\ref{eq:click_prob}\footnotemark \ \;
    Update $\theta$, according to Eq.~\ref{eq: dr_loss}
}

\end{algorithm}
\footnotetext{Due to the large sample space of all possible ranked lists (\ie $N$ documents have $N!$\ permutations), we decrease the sample size in practice.}

%


%% file: 5-exp.tex
\section{Experimental Setup}
To analyze the effectiveness of MULTR, we conduct two types of experiments. The first is simulation experiments based on two of the largest public learning-to-rank datasets. The second is an experiment based on the actual rank lists and user clicks collected from a commercial web search engine. 

\subsection{Simulation Experiment Setup}
To fully investigate the spectrum of propensity estimation and performance of MULTR, we conduct experiments on two of the largest publicly available datasets: 
\begin{itemize}[leftmargin=*]
    \item \textbf{Yahoo! LETOR}\footnote{\url{https://webscope.sandbox.yahoo.com/}} comes from the Learn to Rank Challenge version 2.0 (Set 1), and it is one of the largest benchmark datasets widely used in unbiased learning to rank~\cite{DBLP:conf/sigir/AiBLGC18, DBLP:conf/www/HuWPL19}. It consists of 29,921 queries and 710K documents. Each query-document pair is represented by a 700-D feature vector, and annotated with 5-level relevance labels~\cite{DBLP:journals/jmlr/ChapelleC11}.
    \item \textbf{Istella-S}\footnote{\url{http://quickrank.isti.cnr.it/istella-dataset/}} contains 33K queries and 3,408K documents (roughly 103 documents per query) sampled from a commercial Italian search engine. Each query-document pair is represented by 220 features and annotated with 5-level relevance judgments~\cite{DBLP:conf/sigir/LuccheseNOPST16}.
\end{itemize}
We follow the predefined data split of training, validation and testing of all datasets. The Yahoo!\ set splits the queries arbitrarily and uses 19,944 for training, 2,994 for validation and 6,983 for testing. The Istella-S dataset has been divided into train, validation and test sets according to a $60\% - 20\% - 20\%$ scheme.
    
\paragraph{Click Simulation}
We generate click data on both datasets with a two-step process as in \citet{DBLP:conf/wsdm/JoachimsSS17} and \citet{DBLP:conf/sigir/AiBLGC18}. First, we train a Rank SVM model~\cite{DBLP:conf/kdd/Joachims06} using $1\%$ of the training data with  real relevance labels to generate the initial ranked list $\pi_q$ for each query $q$. Then, we simulate the user browsing process and sample clicks from the initial ranked list by utilizing the simulation model. The simulation model generates clicks based on the examination hypothesis in Equation~\ref{eq:assumption}.
Following the methodology proposed by Chapelle \etal~\cite{DBLP:conf/cikm/ChapelleMZG09}, the relevance probability is  set to be
\begin{equation}
    \Pr (r_{d_i} = 1 | \pi_q) = \epsilon + (1 - \epsilon) \frac{2 ^ y - 1}{2^{y_{\mathrm{max}}} - 1},
\end{equation}
where $y \in [0, y_{\mathrm{max}}]$ is the relevance label of the document $d_i$, and $y_{\mathrm{max}}$ is the maximum value of $y$, which is 4 on both datasets. $\epsilon$~is the noise level, which models click noise such that irrelevant documents (\ie $y = 0$) have a non-zero probability to be perceived as relevant and clicked. We fix $\epsilon = 0.1$ as the default setting. 

We use the cascade model~\cite{DBLP:conf/wsdm/CraswellZTR08} to generate the examination probability. It is a context-aware click model, and a user is modeled as searching and clicking documents from top to bottom, and deciding whether to click each result before moving to the next; users stop examining a search result page after the first click. The cascade model is defined as, 
\begin{equation}
    \begin{aligned}
        \Pr(e_{d_1} = 1) =& 1 \\
        \Pr(e_{d_j} | e_{d_{j-1}} = e, c_{d_{j-1}} = c) =& e \cdot (1-c) 
    \end{aligned}
\end{equation} %
where $e$ and $c$ are binary variables that indicate whether the document $d$ is examined and clicked. 

\paragraph{Baselines}
To demonstrate the effectiveness of our proposed method, we compare with baseline methods, including IPW-based methods and bias modeling methods, which are widely used in ULTR problems.
\begin{itemize}[leftmargin=*]
  \item \textbf{DLA}: The Dual Learning Algorithm~\cite{DBLP:conf/sigir/AiBLGC18} treats the problem of unbiased learning to rank and unbiased propensity estimation as a dual problem, such that they can be optimized simultaneously.  
  \item \textbf{REM}: The Regression EM model~\cite{DBLP:conf/wsdm/WangGBMN18} uses an EM framework to estimate the propensity scores and ranking scores. 
  \item \textbf{PairD}: The Pairwise Debiasing (PairD) Model~\cite{DBLP:conf/www/HuWPL19} uses inverse propensity weighting for pairwise learning to rank. 
  \item \textbf{PAL}: The Position-bias Aware Learning framework~\cite{DBLP:conf/recsys/GuoYLTZ19} is a bias modeling method. It introduces a position model to explicitly represent position bias, and jointly model the bias and relevance. 
  \item \textbf{Oracle}: This model utilizes human experts' annotated labels to train the ranking model and test its performance. Its performance can be considered as an upper bound for the ranking model.
  \item \textbf{Naive}: This model just uses the raw click data to train the ranking model, without any correction. Its performance can be considered as a lower bound for the ranking model.
\end{itemize}
\noindent
In addition, we also consider two variants of MULTR as follows,
\begin{itemize}[leftmargin=*]
    \item \textbf{Rand-MULTR}: We replace the well-trained user simulator with the model with random initialization, \ie we skip lines 2-5 in the Algorithm 1. We refer to this variant
as Rand-MULTR.
    \item \textbf{EIB-MULTR}: We alternate the doubly robust loss in Equation~\ref{eq: dr_loss} with the error imputation based loss~\cite{DBLP:conf/icml/Hernandez-LobatoHG14a}, defined as $\ell_{EIB} = \frac{1}{|\Pi|} \sum_{\pi \in \Pi} \ell_{naive}(f,q|\pi)$. It is called EIB since it uses the user simulator to compute an imputed error, \ie the estimated values of the ranking loss, for each unobserved ranked list, and then used it to estimate the true ranking loss for all the ranked lists. We refer this variant as EIB-MULTR. 
\end{itemize}

\paragraph{Experimental Protocols}
We implement MULTR
and use the baselines in ULTRA\footnote{\url{https://github.com/ULTR-Community/ULTRA_pytorch/}}~\cite{DBLP:conf/cikm/TranYA21} to conduct our experiments. In particular, MULTR is integrated with DLA, as we consider it to be
the best method to estimate propensity. 
For each query, only the top $N=10$ documents are assumed to be displayed to the users. For both datasets, all models are trained with synthetic clicks. Following the setting in~\citet{DBLP:conf/sigir/AiBLGC18}, the click sessions for training are generated on the fly. We fix the batch size to 256 and train each model for 10K steps. 
The user simulator is 
also trained with batch size of 256 and 10K steps; afterwards, we fix it to generate pseudo-click labels. We use the AdaGrad optimizer~\cite{DBLP:journals/jmlr/DuchiHS11} and tune learning rates from
0.01 to 0.05 for each unbiased learning-to-rank algorithm.

In our experiments, we train neural networks for our ranking functions. All reported results are produced using a model with three hidden layers with size $[512, 256, 128]$ respectively, with the ELU~\cite{DBLP:journals/corr/ClevertUH15} activation function and 0.1 dropout~\cite{DBLP:journals/jmlr/SrivastavaHKSS14}.  In terms of imputation model, the dimension of the hidden states is 64. We tune the $L_2$ regularization coefficient $\lambda, \mu$ in range of $\{1e^{-5}, 1e^{-4},\cdots, 1\}$, and the number of pseudo samples in $\{2, 4, 8, 16, 24, 32\}$ based on results in the validation set.

To evaluate all methods, we use the normalized Discounted Cumulative Gain (nDCG)~\cite{DBLP:journals/tois/JarvelinK02} and the Expected Reciprocal Rank (ERR)~\cite{DBLP:conf/cikm/ChapelleMZG09}. For both metrics, we report the results at rank 1, 3, 5, and 10 to show the performance of models at different positions. Following Ai \etal~\cite{DBLP:conf/sigir/AiBLGC18}, statistical differences are computed based on the Fisher randomization test~\cite{DBLP:conf/cikm/SmuckerAC07} with $p \leq 0.05$.

\subsection{Real Click Experiment Setup}
In order to show the effectiveness of MULTR in practice, we also conduct experiments on click data \textbf{Tiangong-ULTR}\footnote{\url{http://www.thuir.cn/data-tiangong-ultr/}} collected from a commercial web search engine \cite{DBLP:conf/cikm/AiMLC18, DBLP:conf/sigir/AiBLGC18}. It contains 3,449 queries written by real search engine users and the corresponding top 10 results are sampled from a two-week search log collected on Sogou\footnote{\url{https://www.sogou.com/}}. The raw HTML documents are downloaded,
and then the rank lists that cannot be crawled are removed. After cleaning, there are 333,813 documents, 71,106 ranked lists and 3,268,177 anonymous click sessions. 

\paragraph{Feature Extraction}
To train the learning-to-rank models, features are extracted based on the text of queries and documents. The ranking features are constructed based on the URL, title, content and whole text of the documents and queries~\cite{DBLP:conf/sigir/AiBLGC18}. In total, each query-document pair has 33 features.

\paragraph{Evaluation}
Tiangong-ULTR provides a test set with 100 queries written by real users. Each query has 100 candidate documents (retrieved by BM25) and each query-document pair has 5-level relevance annotation. We train the ranking model in MULTR and baselines on the training set with clicks following the same protocols in the simulation experiments, and we evaluate their performance on the test set with human annotations. Similar to the simulation experiments, we report nDCG and ERR for all models.

{\centering
\tabcolsep 0.02in
\begin{table*}[!hpt]
\centering
\caption{A comparison of the overall performance MULTR and competing methods on Yahoo! and Istella-S datasets. $``\ast"$  indicates a statistically significant improvement over the best baseline.}
\vspace{-0.3cm}
\label{tab:overall}
    \centering
    \begin{tabular}{ c c c c c  c c c c c  c c c c c  c c}
    \ChangeRT{0.8pt}
    \multirow{3}{*}{\shortstack{Methods}} & \multicolumn{8}{c}{ Yahoo! LETOR } & \multicolumn{8}{c}{ Istella-S } 
    \\ \cmidrule(lr){2-9} \cmidrule(lr){10-17}
    & \multicolumn{4}{c}{ NDCG@K } & \multicolumn{4}{c}{ ERR@K } & \multicolumn{4}{c}{ NDCG@K } & \multicolumn{4}{c}{ ERR@K }  \\
    \cmidrule(lr){2-5}  
    \cmidrule(lr){6-9}      
    \cmidrule(lr){10-13}  
    \cmidrule(lr){14-17} 
    & K = 1 & K = 3 & K = 5 & K = 10 & K = 1 & K = 3 & K = 5 & K = 10  & K = 1 & K = 3 & K = 5 & K = 10 & K = 1 & K = 3 & K = 5 & K = 10 \\
    \ChangeRT{0.5pt} 
    Oracle & 0.689 & 0.691 & 0.714 & 0.761 & 0.350 & 0.428 & 0.449 & 0.465 
           & 0.671 & 0.641 & 0.667 & 0.729 & 0.598 & 0.709 & 0.725 & 0.731 \\
    \ChangeRT{0.5pt}
    \bf{MULTR}  & $\textbf{0.681}^\ast$ & $\textbf{0.682}^\ast$ & $\textbf{0.705}^\ast$ & $\textbf{0.751}^\ast$
           & $\textbf{0.348}^\ast$ & $\textbf{0.426}^\ast$ & $\textbf{0.447}^\ast$ & $\textbf{0.463}^\ast$ 
           & $\textbf{0.664}^\ast$ & $\textbf{0.627}^\ast$ & $\textbf{0.648}^\ast$ & $\textbf{0.704}^\ast$
           & $\textbf{0.594}^\ast$ & $\textbf{0.701}^\ast$ & $\textbf{0.717}^\ast$ & $\textbf{0.724}^\ast$ \\ 
    DLA    & 0.669 & 0.671 & 0.694 & 0.743 & 0.346 & 0.422 & 0.444 & 0.460 
           & 0.639 & 0.607 & 0.629 & 0.683 & 0.572 & 0.683 & 0.700 & 0.707 \\
    REM    & 0.636 & 0.638 & 0.661 & 0.711 & 0.336 & 0.410 & 0.432 & 0.448 
           & 0.590 & 0.551 & 0.568 & 0.618 & 0.528 & 0.640 & 0.659 & 0.669 \\
    PairD  & 0.653 & 0.662 & 0.687 & 0.738 & 0.333 & 0.413 & 0.436 & 0.451 
           & 0.604 & 0.586 & 0.617 & 0.687 & 0.537 & 0.661 & 0.680 & 0.688 \\
    PAL    & 0.645 & 0.660 & 0.683 & 0.735 & 0.328 & 0.410 & 0.432 & 0.448 
           & 0.620 & 0.595 & 0.623 & 0.691 & 0.553 & 0.672 & 0.690 & 0.698 \\
    Naive  & 0.651 & 0.660 & 0.685 & 0.737 & 0.332 & 0.411 & 0.434 & 0.450 
           & 0.613 & 0.592 & 0.622 & 0.690 & 0.546 & 0.667 & 0.686 & 0.694 \\ \hline 
    Rand-MULTR  & 0.667 & 0.670 & 0.693 & 0.743 & 0.341 & 0.419 & 0.441 & 0.457 
               & 0.633 & 0.604 & 0.631 & 0.693 & 0.564 & 0.681 & 0.698 & 0.706 \\
    EIB-MULTR   & 0.674 & 0.678 & 0.700 & 0.749 & 0.346 & 0.423 & 0.445 & 0.460 
               & 0.653 & 0.611 & 0.631 & 0.684 & 0.585 & 0.690 & 0.707 & 0.714\\ 
    \ChangeRT{0.8pt}
    \end{tabular}
\end{table*}}

\section{Results and Analysis}

In this section, we discuss the results of our proposed model-based unbiased learning to rank method with existing approaches using both simulated and real-world experiments. In general, we expect the experimental results to answer the following research questions:
\begin{itemize}[leftmargin=*]
    \item \textbf{RQ1}: Can MULTR outperform state-of-the-art unbiased learning to rank methods?
    \item \textbf{RQ2}: What influence do variant designs have on MULTR?
    \item \textbf{RQ3}: How does the sample number of unobserved ranked lists influence the performance of MULTR?
    \item \textbf{RQ4}: How does MULTR perform on real click logs?
\end{itemize}

\subsection{Performance Comparison (RQ1)}
To answer RQ1, we compare MULTR with other unbiased learning-to-rank algorithms in the simulation experiments. Table~\ref{tab:overall} summarizes the performance of different unbiased learning to rank algorithms under various click situations. 
We find that:
\begin{itemize}[leftmargin=*]
    \item Our model-based unbiased learning to rank achieves the best performance among all the state-of-the-art methods in terms of nDCG and ERR, which indicates our model is robust when propensity cannot be accurately estimated. 
    \item Reducing variance can enhance performance on unbiased learning to rank. As MULTR shares the same method with DLA in estimating propensity, the improvements over DLA demonstrate the necessity of handling variance in IPW-based methods.
    \item The naive method may achieve better performance than some unbiased learning to rank methods when clicks are generated by a cascade model. One explanation is that documents displayed at a lower rank have a higher opportunity to be examined; thus they receive more clicks, which alleviates the bias in click data. This observation confirms that a mismatch between propensity estimation and real bias could lead to performance even worse than raw click data. 
    \item The oracle model with human annotation consistently achieves the best performance. It implies that there is still room to improve in unbiased learning to rank methods. 
\end{itemize}

\subsection{Ablation Study (RQ2)}
MULTR has specific design features, including a context-aware user simulator and a doubly robust loss function to train rankers. To demonstrate their effectiveness, we analyze their respective impacts on the model's performance via ablation study. The experimental results of its two variants on two datasets are summarized in Table~\ref{tab:overall}. We analyze their respective effects as follows.

(1) Rand-MULTR: When we replace the well-trained user simulator with one from random initialization, the performance of Rand-MULTR significantly degrades.  This observation verifies the necessity of an accurate user simulator for the doubly robust estimator learning process: when the user simulator has a high bias due to inaccurate click imputations, it will mislead the rankers.

(2) EIB-MULTR: When we use the same treatments for real clicks from observations and pseudo clicks from user simulators, the performance of EIB-MULTR suffers a significant decrease. This suggests, naturally, that there is a discrepancy between real clicks and pseudo clicks. The doubly robust approach effectively leverages pseudo clicks and improves the performance of rankers.

\subsection{Parameter Sensitivity Study (RQ3)}
As we have explained previously, when $N$ is the maximum number of documents that can be shown to users, there will be $N!$ candidate ranked lists, which could be large. Therefore, we decrease the number of samples for unobserved ranked lists in practice. To investigate its impact in MULTR, we vary the sample size for unobserved ranked lists in the range of $\{0, 2, 4, 8, 16, 24, 32\}$. Figure ~\ref{fig:sample_ratio} shows the NDCG@K and ERR@K for MULTR with respect to different sample sizes on both datasets.

\begin{figure*}[h!]
    \centering
    \includegraphics[width=6.8in]{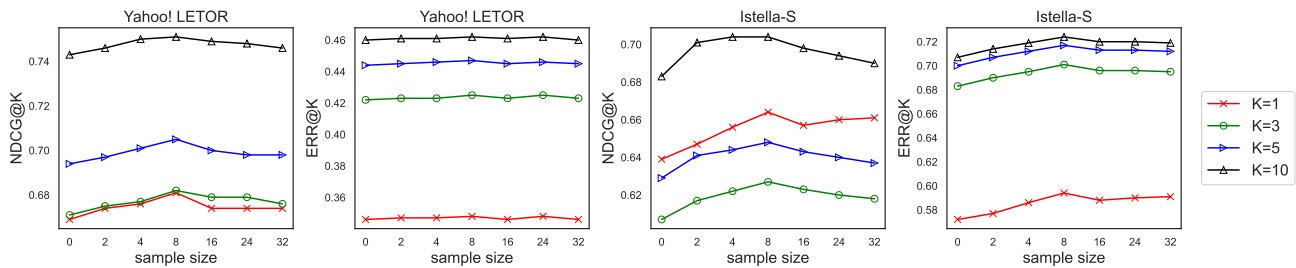}
    \caption{Study of effect of sample size for unobserved ranked lists on Yahoo! LETOR and Istella-S.}
    \label{fig:sample_ratio}
\end{figure*}

First, MULTR with sample number 0, \ie we merely sample from unobserved candidate rank lists, derives the worst performance. This observation shows that the well-trained user simulator enables the unobserved rank lists to provide rankers with useful information. Second, when the sampling number increases from 0 to 8, the performance at all levels increases. This shows that sampling more unobserved rank lists can bring more beneficial information to rankers.  Third, overly large sampling does not show a significant improvement, even though we should sample all unobserved rank lists in theory. We even observe that NDCG starts to decrease when the sample number increases from 8 to 32. One possible reason could be that the observed rank lists are typically sparse in practice, meaning that we cannot ensure that the rankers obtain sufficient information. For both datasets, the optimal sampling size is 8. In conclusion, setting the sampling size too conservatively or too aggressively may adversely affect the ranking performance, and too aggressively also brings in unnecessary computational cost.

\subsection{Real Data Study (RQ4)}
As we have demonstrated that the proposed MULTR works well in the simulation study, we would further investigate how it performs on real click logs recorded by the search engine. We compare MULTR with existing unbiased learning to rank algorithms on real data. The experimental results are summarized in Table~\ref{tab:real-exp}. 
From the results, we can see that MULTR outperforms state-of-the-art unbiased learning to rank algorithms. 
This shows that the user simulator in MULTR provides reliable pseudo-clicks for the ranker,which enhances the ranking performance.
In addition, propensity models in both MULTR and baseline methods are developed upon position-based examination assumption, which makes accurate propensity estimation infeasible on real data. Therefore, they are more likely to suffer from the high variance problem. The superior performance of MULTR over baseline methods verifies the necessity of reducing the variance via doubly-robust learning.

\begin{table}[!hpt]
    \caption{Comparison of MULTR and competing methods on Tiangong-ULTR. }
    \vspace{-0.3cm}
    \tabcolsep 0.04in
    \small
    \centering
    \begin{tabular}{ c c c c c  c c c c}
    \ChangeRT{0.8pt}
    \multirow{2}{*}{\shortstack{Methods}}  & \multicolumn{4}{c}{ NDCG@K } & \multicolumn{4}{c}{ ERR@K } \\
    \cmidrule(lr){2-5} \cmidrule(lr){6-9}
    & K = 1 & K = 3 & K = 5 & K = 10 & K = 1 & K = 3 & K = 5 & K = 10 \\
    \ChangeRT{0.5pt} 
    \bf{MULTR}   & $\textbf{0.483}$ & $\textbf{0.465}$ & $\textbf{0.469}$ & $\textbf{0.476}$
            & $\textbf{0.448}$ & $\textbf{0.578}$ & $\textbf{0.606}$ & $\textbf{0.618}$ \\
    DLA     & 0.472 & 0.443 & 0.425 & 0.443 & 0.443 & 0.571 & 0.594 & 0.610 \\
    REM     & 0.450 & 0.455 & 0.457 & 0.476 & 0.422 & 0.566 & 0.594 & 0.606 \\
    PairD   & 0.388 & 0.371 & 0.369 & 0.383 & 0.364 & 0.494 & 0.522 & 0.542 \\
    PAL     & 0.365 & 0.384 & 0.395 & 0.413 & 0.342 & 0.483 & 0.516 & 0.533 \\
    Naive   & 0.412 & 0.413 & 0.408 & 0.411 & 0.386 & 0.528 & 0.556 & 0.570 \\
    \ChangeRT{0.8pt}
    \end{tabular}
    \label{tab:real-exp}
\end{table}

%% file: 6-conclusion.tex
\vspace{-0.5cm}
\section{Conclusion}

In this work, we propose MULTR, a model-based unbiased learning to rank framework. 
We first design a context-aware user simulator to produce labels for unobserved ranked lists as supervision signals, which addressed data sparsity for long-tail queries. However, the natural discrepancy between pseudo and actual clicks could mislead the ranking models, as pseudo data cannot be as accurate as real users.  
Therefore, we resort to DR approaches. In particular, we take the observation of the ranked lists as the treatment, which addresses the treatment unobservation when naively applying existing DR methods. Afterward, we take the pseudo-click data as data imputation, and propose a doubly-robust learning algorithm to obtain unbiased ranking models. 
Theoretical analysis reveals that our method is more robust than existing methods. Extensive experiments on benchmark datasets, including simulated datasets and real click logs, demonstrate that the proposed model-based method consistently outperforms state-of-the-art methods. 